\documentclass[preprint]{ptephy_v1}

\preprintnumber{arXiv:2109.14485} 
\usepackage{hyperref}
\usepackage{graphicx}
\usepackage{dcolumn}
\usepackage{bm}

\usepackage[utf8]{inputenc}
\usepackage[T1]{fontenc}
\usepackage{mathptmx}
\usepackage{etoolbox}
\usepackage{graphicx}
\usepackage{dcolumn}
\usepackage{bm}
\usepackage{amssymb}   
\usepackage{color}
\usepackage{amsmath}
\usepackage{amsmath,amscd}
\usepackage{caption,subcaption}
\usepackage{here}
\usepackage{amsthm}
\newtheorem{theorem}{Theorem}

\newtheorem{definition}[theorem]{Definition}

\begin{document}

\title{Universal Critical Behavior of Transition to Chaos: Intermittency Route}


\author{Ken-ichi Okubo\thanks{Present Address: Department of Information and Physical Sciences, 
		Graduate School of Information Science 
		and Technology, Osaka University
		1-5 Yamada-oka, Suita, Osaka 565-0871, Japan}}
\author{Ken Umeno}
\affil{Department of Applied Mathematics and Physics, 
	Graduate School of Informatics, Kyoto University, Yoshida-honmachi, Sakyo-ku, Kyoto 606-8501 Japan \email{okubo@ist.osaka-u.ac.jp}}


\begin{abstract}%
The robustness of the universality class concept of the chaotic transition was investigated by analytically obtaining its critical exponent for a wide class of maps.
In particular, we extended the existing one-dimensional chaotic maps, thereby generalising the invariant density function from the Cauchy distribution by adding one parameter. This generalisation enables the adjustment of the power exponents of the density function and {superdiffusive behavior}. We proved that these generalised one-dimensional chaotic maps are exact (stronger condition than ergodicity) to obtain the critical exponent of the Lyapunov exponent from the phase average.
Furthermore, we proved that the critical exponent of the Lyapunov exponent is $\frac{1}{2}$ regardless of the power exponent of the density function and is thus universal. This result can be considered as rigorous proof of the universality of the critical exponent of the Lyapunov exponent for a countably infinite number of maps.
\end{abstract}

\subjectindex{A30, A33, A40}

\maketitle

\section{Introduction}
In the field of nonlinear and statistical physics, various theoretical and experimental research have been performed on the transition from a stable state to a chaotic state \cite{liu2002transition,lai2017quasiperiodicity, cosenza2010lyapunov,liu2003universal,milosavljevic2017analytic,ott1994blowout, benettin1984power, miller1988stochastic,  huberman1980scaling,lamba1994scaling,pomeau1980intermittent,manneville1979intermittency,hioe1987stability,swinney1983observations, ono1995critical, feng1996type,aizawa1984statistical}. Various routes to chaos have been identified considering the factors leading to chaos \cite{huberman1980scaling,pomeau1980intermittent,swinney1983observations}. Notably, the Lyapunov exponent $\lambda$ of the system above the critical point from the stable state to the chaotic state with parameter $a$ can be approximately expressed as
\begin{equation*}
    \lambda \sim b|a-a_c|^\nu,
\end{equation*}
where $b$ is a constant, $a_c$ is the critical point, and $\nu$ is the critical exponent. 
The critical exponent $\nu$ can be considered to be the crux indicator of phase transition between predictable and chaotic states as in other critical exponents of phase transition in physics.
Moreover, this exponent is expected to be universal and constant for each route to chaos regardless of minor differences among systems \cite{pomeau1980intermittent}. 
Particularly, Huberman and Rudnick performed a numerical simulation to estimate that for the period-doubling route, $\nu$ in the logistic map is $\frac{\log 2}{\log \delta}=0.4498069...$, where $\delta=4.669201...$ is the Feigenbaum constant\cite{huberman1980scaling}.
The intermittency route is another route. Pomeau and Manneville \cite{pomeau1980intermittent} estimated that {$\nu=1/2$} for intermittency types I and III. Recently, machine learning has been used to predict the critical phenomena of chaos \cite{kong2021machine}.

Moreover, researchers have rigorously proved that $\nu=1/2$ by demonstrating the ergodic property for one-parameter maps, which exhibit intermittency and preservation of the Cauchy distribution \cite{umeno2016exact,okubo2018universality}. 

Cauchy distribution has a fat tail in which the expectation value and the variance cannot be defined.
This fat tail can be observed in a lot of fields, including earthquakes \cite{utsu1995centenary,utsu1999representation}, the worldwide web \cite{albert1999diameter}, economics and finance \cite{gabaix2009power,gopikrishnan2000statistical}, and chaotic dynamics \cite{altmann2007hypothesis}.
According to the generalised central limit theorem (GCLT) \cite{gnedenko1954limit}, the limit distribution of summation of random numbers, which obey an arbitrary distribution on the domain of attraction characterised by the power exponent $\beta + 1$ such as $\rho(x) \approx |x|^{-(\beta+1)}$, converges to a stable distribution with exponent $\beta$. Therefore, the stable distribution with exponent $\beta$ can be regarded as a fixed point of the universality class of the renormalisation group in the domain of attraction characterised by $\beta$.

In this paper, we consider a system in which the class of invariant density functions is extended by adding a new parameter to the existing maps. Here, the class of {superdiffusive character} (power exponent of the invariant density function) can be changed by adding the parameter.
We prove that the critical exponent of the Lyapunov exponent is $\nu=1/2$ regardless of the power exponent by demonstrating the ergodic property although the universality class in the {GCLT} can
be changed by adjusting the added parameter $\beta$.
In other words, the universality of the critical exponent of the Lyapunov exponent is \textit{more robust} than that of the class of power exponents.

\section{Dynamics}
An extension of super generalised Boole (SGB) transformations \cite{okubo2018universality} is considered.
For SGB transformations, the Cauchy distribution is determined as a unique invariant density when appropriate conditions are satisfied.
Adopting the method of a previous study \cite{umeno1998superposition}, parameter $\beta$ is added in SGB transformations to extend the class of the invariant density.
The extended maps are described in the following.
Let $F_K(x)$ be the K-angle formula of the cot function such that
\begin{equation}
F_K(\cot \theta) \overset{\mathrm{def}}{=} \cot K\theta.
\end{equation}
For example, $F_3, F_4$, and $F_5$ are expressed as
\begin{equation*}
    \begin{array}{lll}
        F_3(x) &=& \displaystyle \frac{x^3-3x}{3x^2-1},\\
        F_4(x) &=& \displaystyle \frac{x^4-6x^2+1}{4x^3-4x},\\
        F_5(x) &=& \displaystyle \frac{x^5-10x^3+5x}{5x^4-10x^2+1}.
    \end{array}
\end{equation*}
Subsequently, function $U_{K, \alpha, \beta}(x)$ is defined as
\begin{equation}
U_{K,\alpha, \beta}(x) 
\overset{\mathrm{def}}{=} \left|\alpha K F_K\left(|x|^\beta\mbox{sgn}(x)\right)\right|^{\frac{1}{\beta}}\mbox{sgn}
\left\lbrace \alpha F_K\left(|x|^\beta\mbox{sgn}(x)\right)\right\rbrace,
\end{equation}
where $|\alpha|>0$ and $0<\beta<2$. 
The extended super generalised Boole (ESGB) transformation is defined as
\begin{equation}
x_{n+1} = U_{K,\alpha, \beta}(x_n),
\end{equation}
where SGB transformations correspond to the case with $\beta=1$.
For example, $U_{3,\alpha, \beta}, U_{4,\alpha, \beta}$, and $U_{5,\alpha, \beta}$ are given by
\begin{equation}
\begin{array}{lll}
     x_{n+1} = U_{3,\alpha, \beta}(x_n) &=& \displaystyle\left|3\alpha \frac{y^3-3y}{3y^2-1}\right|^\frac{1}{\beta}\mbox{sgn}
     \left\lbrace \alpha \frac{y^3-3y}{3y^2-1}\right\rbrace,\\
     x_{n+1} = U_{4,\alpha, \beta}(x_n) &=& \displaystyle\left|4\alpha \frac{y^4-6y^2+1}{4y^3-4y}\right|^\frac{1}{\beta} \mbox{sgn}
     \left\lbrace \alpha \frac{y^4-6y^2+1}{4y^3-4y} \right\rbrace,\\
     x_{n+1} = U_{5,\alpha, \beta}(x_n) &=& \displaystyle\left|5\alpha \frac{y^5-10y^3+5y}{5y^4-10y^2+1}\right|^\frac{1}{\beta}
     \mbox{sgn}\left\lbrace \alpha \frac{y^5-10y^3+5y}{5y^4-10y^2+1}\right\rbrace,
\end{array}
\end{equation}
where $y = |x_n|^\beta \mbox{sgn}(x_n)$.
Figure \ref{Fig: Form of ESGB} illustrates the behavior of the return maps of $U_{3,\frac{1}{3},\frac{1}{2}}$, $U_{4,\frac{1}{4},\frac{1}{2}}$, and $U_{5,\frac{1}{5},\frac{1}{2}}$.
Moreover, Figure \ref{Fig: Form of ESGB3-change-beta} shows $U_{3,\frac{1}{3},\frac{1}{2}}$, $U_{3,\frac{1}{3},1}$, and $U_{3,\frac{1}{3},\frac{3}{2}}$.
\begin{figure}[H]
	\centering
	\includegraphics[width = .6\columnwidth]{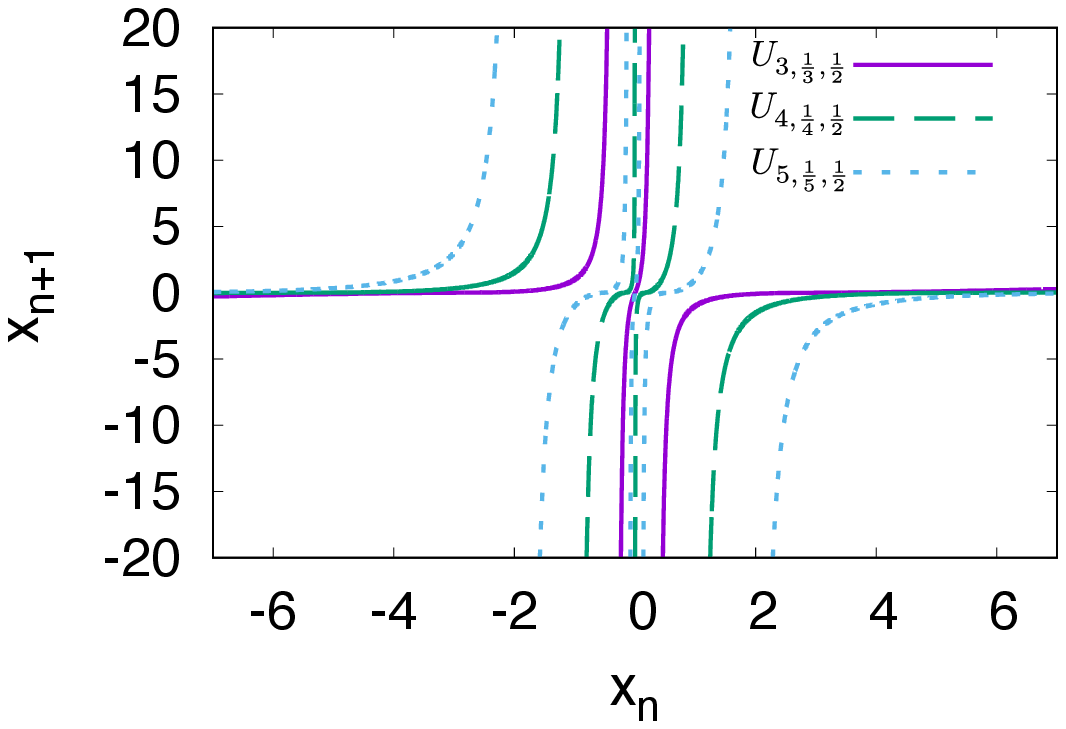}
	\caption{Return maps of $U_{3,\frac{1}{3},\frac{1}{2}}$, $U_{4,\frac{1}{4},\frac{1}{2}}$, and $U_{5,\frac{1}{5},\frac{1}{2}}$.
		The solid, broken, and dotted lines correspond to $U_{3,\frac{1}{3},\frac{1}{2}}$, $U_{4,\frac{1}{4},\frac{1}{2}}$, and $U_{5,\frac{1}{5},\frac{1}{2}}$,
		respectively.}
	\label{Fig: Form of ESGB}
\end{figure}

\begin{figure}[H]
	\centering
	\includegraphics[width = .6\columnwidth]{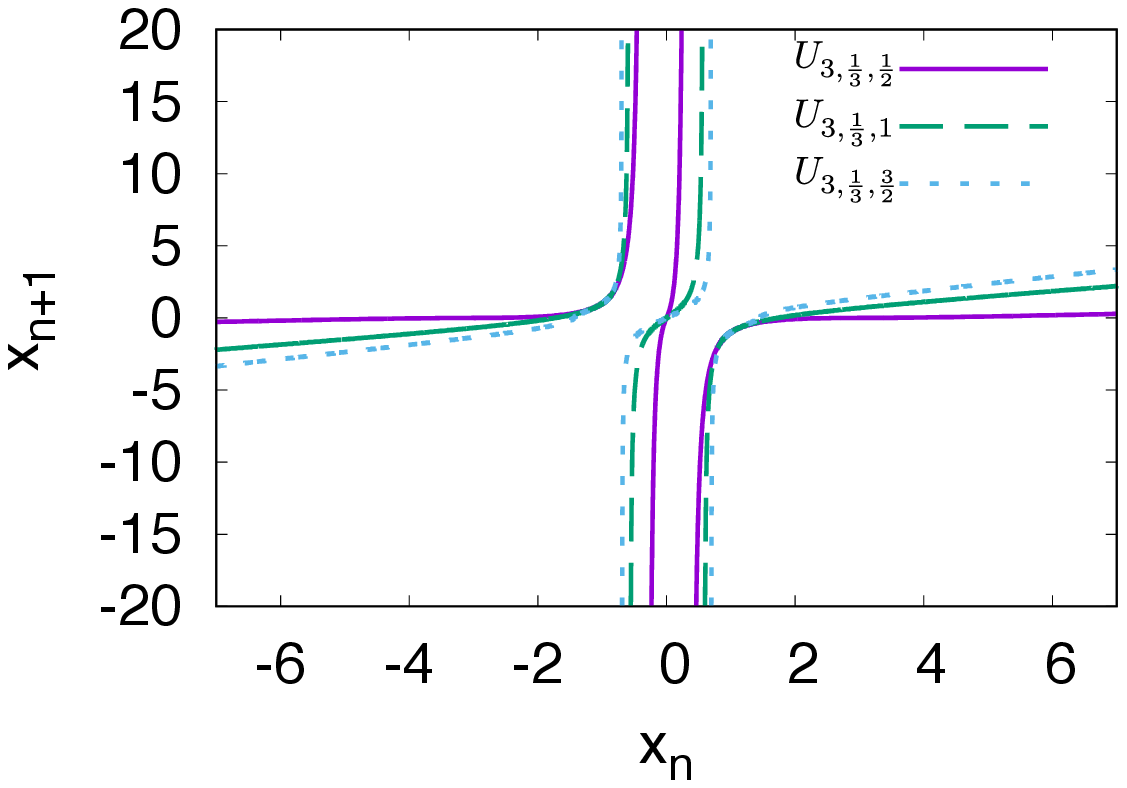}
	\caption{Return maps of $U_{3,\frac{1}{3},\frac{1}{2}}$, $U_{3,\frac{1}{3},1}$, and $U_{3,\frac{1}{3},\frac{3}{2}}$.
		The solid, broken, and dotted lines correspond to $U_{3,\frac{1}{3},\frac{1}{2}}$, $U_{3,\frac{1}{3},1}$, and $U_{3,\frac{1}{3},\frac{3}{2}}$,
		respectively.}
	\label{Fig: Form of ESGB3-change-beta}
\end{figure}

The map $U_{K, \alpha, \beta}$ is a $K$-to-one map; this is detailed {in the Appendix}.
{Consider} points $x_j, j= 1, \cdots, K$ on $\mathbb{R}$ as follows:
\begin{equation}
x_j=\left|\cot\left(\theta+j\frac{\pi}{K}\right)\right|^\frac{1}{\beta}\mbox{sgn}\left(\cot\left(\theta+j\frac{\pi}{K}\right)\right), j= 1,\cdots, K.
\end{equation}
In this case, $|x_j|^\beta$ and $|x_j|^\beta\mbox{sgn}(x)$ can be expressed as
\begin{equation}
\begin{array}{lll}
|x_j|^\beta &=& \left\lbrace 
\begin{array}{rl}
\cot\left(\theta+j\frac{\pi}{K}\right), & \cot\left(\theta+j\frac{\pi}{K}\right) \geq 0,\\
-\cot\left(\theta+j\frac{\pi}{K}\right), & \cot\left(\theta+j\frac{\pi}{K}\right)<0,
\end{array}
\right.\\
|x_j|^\beta\mbox{sgn}(x_j) &=& \cot\left(\theta+j\frac{\pi}{K}\right).
\end{array}
\end{equation}
The following equation can be obtained by transferring $x_j$ from $U_{K, \alpha, \beta}$.
\begin{equation}
U_{K,\alpha, \beta}(x_j) = \left|\alpha K\cot(K\theta)\right|^\frac{1}{\beta}\mbox{sgn}\left\lbrace \alpha\cot(K\theta)\right\rbrace. 
\end{equation}
Subsequently, we define $y=\left|\alpha K\cot(K\theta)\right|^\frac{1}{\beta}\mbox{sgn}\left\lbrace \alpha\cot(K\theta)\right\rbrace$. Note that 
there exist $K$ pieces $x_j, j=1,\cdots, K$, satisfying the relation $y=U_{K, \alpha, \beta}(x_j)$.

\section{Invariant density}
In this section, the condition in which the ESGB transformation has an invariant density is examined.
A density function
\begin{equation}
\rho_n(x) = \frac{1}{\pi}\frac{\gamma |x|^{\beta-1}}{\gamma^2+|x|^{2\beta}} \label{Eq: density}
\end{equation}
is known a priori. 
Based on the discussion in Section II, the map $U_{K, \alpha, \beta}$ is a $K$-to-one map.
Using the Perron--Frobenius operator, $\rho_{n+1}(y)$ can be obtained.
\begin{equation}
\rho_{n+1}(y)|dy| = \rho_n(x_{0})|dx|_{x=x_{0}}+\cdots+\rho_n(x_{K-1})|dx|_{x=x_{K-1}},
\end{equation}
where 
$x_j=\left|\cot\left(\theta+j\frac{\pi}{K}\right)\right|^{\frac{1}{\beta}}\mbox{sgn}\left\lbrace \cot\left(\theta+j\frac{\pi}{K}\right)\right\rbrace, 
j=0,\cdots,K-1$ and $y=\left|\alpha K \cot\left(K\theta\right)\right|^{\frac{1}{\beta}}\mbox{sgn}\left\lbrace \cot(K\theta)\right\rbrace $.
In this case,
\begin{equation}
\begin{array}{lll}
\displaystyle\left|\frac{dy}{dx}\right|_{x=x_j} &=& |\alpha|^\frac{1}{\beta}K^{\frac{1}{\beta}+1}\displaystyle \frac{\left|\cot(K\theta)\right|^{\frac{1}{\beta}-1}\sin^2\left(\theta+j\frac{\pi}{K}\right)}
{\left|\cot\left(\theta+j\frac{\pi}{K}\right)\right|^{\frac{1}{\beta}-1}\sin^2(K\theta)},\\
&=& |\alpha|^\frac{1}{\beta}K^{\frac{1}{\beta}+1}\displaystyle \frac{\left|\cot(K\theta)\right|^{\frac{1}{\beta}-1}\left(1+\frac{|y|^{2\beta}}{|\alpha K|^2}\right)}
{\left|\cot\left(\theta+j\frac{\pi}{K}\right)\right|^{\frac{1}{\beta}-1}\left(1+|x_j|^{2\beta}\right)}.
\end{array}
\end{equation}
Subsequently, $\rho_{n+1}(y)$ is given by
\begin{equation}
\begin{array}{lll}
\rho_{n+1}(y) &=& \displaystyle \sum_{j=1}^{K} \rho_n(x_j)\left|\frac{dx}{dy}\right|_{x=x_j}\\
&=& \displaystyle\frac{1}{|\alpha|^\frac{1}{\beta}K^{\frac{1}{\beta}+1}\left|\cot(K\theta)\right|^{\frac{1}{\beta}-1}\left(1+\frac{|y|^{2\beta}}{|\alpha K|^2}\right)}
\sum_{j=1}^K\rho_n(x_j)\left|\cot\left(\theta+j\frac{\pi}{K}\right)\right|^{\frac{1-\beta}{\beta}}\left(1+|x_j|^{2\beta}\right).
\end{array}
\end{equation}
In this case \cite{umeno2016ergodic}, the following expressions hold:
\begin{equation}
\begin{array}{lll}
\displaystyle \sum_{j=1}^K\rho_n(x_j)\left|\cot\left(\theta+j\frac{\pi}{K}\right)\right|^{\frac{1-\beta}{\beta}}\left(1+|x_j|^{2\beta}\right)
&=& \displaystyle\frac{1}{\pi}\sum_{j=1}^K\frac{\gamma}{1+(\gamma^2-1)\sin^2\left(\theta+j\frac{\pi}{K}\right)}\\
&=& \displaystyle\frac{1}{\pi}\frac{K G_K(\gamma)\left(1+\cot^2(K\theta)\right)}{G_K^2(\gamma)+\cot^2(K\theta)},
\end{array}
\end{equation}
where $G_K(x)$ corresponds to the K-angle formula of the coth function, which is expressed as 
\begin{equation*}
    G_K(\coth \theta) \overset{\mathrm{def}}{=} \coth K\theta.
\end{equation*}
Therefore, 
\begin{equation}
\rho_{n+1}(y) = \frac{1}{\pi} \frac{|\alpha K| G_K(\gamma)|y|^{\beta-1}}{|\alpha K|^2G_K^2(\gamma)+|y|^{2\beta}}.
\end{equation}
In one iteration, although the form of the density function does not change, the scale parameter can be mapped as
\begin{equation}
\gamma \to |\alpha| K G_K(\gamma).
\end{equation}
The necessary condition that the density associated with \eqref{Eq: density} is the invariant density is expressed as
\begin{equation}
\gamma = |\alpha| K G_K(\gamma). \label{Eq: beta parameter condition}
\end{equation}
Equation \eqref{Eq: beta parameter condition} does not depend on $\beta$. Thus, we can apply the condition reported in the previous study \cite{okubo2021infinite}.
\begin{definition}
	When parameters $(K, \alpha)$ satisfy the following condition: 
	\begin{equation}
	\left\lbrace
	\begin{array}{lll}
	0 < |\alpha| < 1 & \mbox{in the case of} & K = 2N,\\
	\frac{1}{K^2} < |\alpha| < 1 & \mbox{in the case of} & K=2N+1,
	\end{array}
	 \right.
	\end{equation}
	where $N \in \mathbb{N}$ and parameters $(K, \alpha)$ are in Range B.
\end{definition}

\begin{theorem}
	The ESGB transformations $U_{K, \alpha, \beta}$ preserve the density function if parameters $(K, \alpha)$ are in Range B, which is expressed as
	\eqref{Eq: density}.
\end{theorem}

\begin{proof}
    When parameters $(K, \alpha)$ are in Range A, they are defined as
    \begin{equation}
	\left\lbrace
	\begin{array}{lll}
	0 < \alpha < 1 & \mbox{in the case of} & K = 2N,\\
	\frac{1}{K^2} < \alpha < 1 & \mbox{in the case of} & K=2N+1,
	\end{array}
	 \right.
	\end{equation}
	\eqref{Eq: beta parameter condition} has the unique solution $\gamma^*$ \cite{okubo2018universality}.
	Range A' is defined as
	\begin{equation}
	\left\lbrace
	\begin{array}{lll}
	-1 < \alpha < 0 & \mbox{in the case of} & K = 2N,\\
	-1 < \alpha < -\frac{1}{K^2} & \mbox{in the case of} & K=2N+1.
	\end{array}
	 \right.
	\end{equation}
	For Range A', by extending $\alpha$ to $|\alpha|$, a proof can be provided using a similar approach as in the existing studies \cite{okubo2018universality,okubo2021infinite}.
\end{proof}

\section{Exactness}
The ESGB transformations are exact when parameters $(K, \alpha)$ are in Range B. The following presents the proof of exactness.

\begin{theorem}
	ESGB transformations $U_{K, \alpha, \beta}$ are exact if parameters $(K, \alpha)$ are in Range B.
\end{theorem}

\begin{proof}
	The variable $x_n$ is defined as
	\begin{equation}
	x_n = \left| \cot(\pi\theta_n)\right|^{\frac{1}{\beta}}\mbox{sgn}\left\lbrace \cot(\pi\theta_n)\right\rbrace. 
	\end{equation}
	Subsequently, ESGB transformations can be expressed as
	\begin{equation}
	U_{K, \alpha, \beta}(x_n) = \left|\alpha K \cot(\pi K\theta_n)\right|^{\frac{1}{\beta}}\mbox{sgn}\left\lbrace
	\alpha\cot(\pi K \theta_n)\right\rbrace.
	\end{equation}
	Therefore, the following expression can be derived as
	\begin{equation}
	\begin{array}{ccc}
		x_{n+1} &=& \left| \cot(\pi\theta_{n+1})\right|^{\frac{1}{\beta}}\mbox{sgn}\left\lbrace\cot(\pi\theta_{n+1})\right\rbrace\\
		&=& \left|\alpha K \cot(\pi K\theta_n)\right|^{\frac{1}{\beta}}\mbox{sgn}\left\lbrace\alpha\cot(\pi K \theta_n)\right\rbrace.
	\end{array}
	\end{equation}
	Since the signs of both sides are consistent, 
	\begin{equation}
	\mbox{sgn}\left\lbrace \cot(\pi \theta_{n+1})\right\rbrace = \mbox{sgn}\left\lbrace \alpha\cot(\pi K\theta_n)\right\rbrace.  
	\end{equation}
	For $K>0$, we define a map 
	$\Bar{S}_{K,\alpha}$ such that
	\begin{equation}
	\theta_{n+1} = \Bar{S}_{K,\alpha}(\theta_n)=\frac{1}{\pi}\mbox{arccot}\left\lbrace \alpha K \cot(\pi K \theta_n)\right\rbrace. \label{Eq: exact}
	\end{equation}
	Equation \eqref{Eq: exact} is an extension of the equation which appears in the proof part in
	an existing study \cite{okubo2018universality}. Therefore, the map $U_{K, \alpha, \beta}$ is proven exact if $(K, \alpha)$ are in Range B because 
    the  argument to prove the exactness can be applied\cite{okubo2018universality,okubo2021infinite}.
\end{proof}
	{The initial density function converges to Eq. (\ref{Eq: density}) when parameters $(K, \alpha)$ are in Range B.
	(At $\beta=1$, an initial density function converges to the Cauchy distribution, which is a stable distribution.)}

\clearpage
\section{{L\'evy Flight}}
In ESGB transformations, the power exponent of the invariant density function $\rho_{\gamma,\beta}(x)=\frac{1}{\pi}\frac{\gamma|x|^{\beta-1}}{|x|^{2\beta}+\gamma^2}$ can be changed.
For $\beta=1$, we obtained the Cauchy distribution corresponding to the invariant density of SGB transformations.

In the range in which $|x|$ is adequately large, 
\begin{equation}
\rho_{\gamma,\beta} \simeq 
\left\lbrace
\begin{array}{ccc}
{\frac{\gamma}{\pi}}|x|^{-(\beta+1)} &\mbox{for} & x \to \infty,\\
{\frac{\gamma}{\pi}}|x|^{-(\beta+1)} &\mbox{for} & x \to -\infty.
\end{array}
\right.
\end{equation}
Moreover, according to the GCLT \cite{gnedenko1954limit}, 
for the independent and identically distributed random variable $\{x_k\}_{k=1}^\infty$ that follows the density function $\rho_{\gamma,\beta}$, $\frac{\sum_{k=1}^nx_k-A_n}{n^{\frac{1}{\beta}}}$ converges to a stable distribution {with the characteristic exponent $\beta$}\cite{umeno1998superposition}.
$A_n$ is represented as
\begin{equation*}
    A_n = \left\lbrace
\begin{array}{cc}
0, &  0< \beta <1,\\
n^2\mathfrak{F}\log\left(\phi_x\left(\frac{1}{x}\right)\right), & \beta=1,\\
n\mathbb{E}[x], & 1 < \beta < 2,
\end{array}    
    \right.
\end{equation*}
where $\phi_x$ is the characteristic function of $x$.
In this case, {as $\frac{\sum_{k=1}^n x_k -A_n}{n^{1/\beta}}= O(1)$ for large $n$,
\begin{equation}
\begin{array}{ccc}
\displaystyle \sum_{k=1}^n x_k &=& O\left(n^\frac{1}{\beta}\right),  \\
\displaystyle \left(\sum_{k=1}^n x_k\right)^2 & = & O\left(n^{\frac{2}{\beta}}\right). 
\end{array}
\end{equation}
}
For the following dynamics, $p_{n+1}=p_n + x_n$, if we calculate the {mean} square displacement {$\langle(p_n-p_0)^2\rangle = \langle\left(\sum_{k=0}^{n-1}x_k\right)^2\rangle$},
its order obeys $O\left(n^{\eta}\right), \eta=\frac{2}{\beta}, {1\leq}\beta<2$, {where the ensemble is taken from within a bounded region except for the infinity point}. {Therefore, this L\'evy flight represented by the random walk $p_n =\sum_{i=0}^{n-1} x_i+p_0$ exhibits the \textit{superdiffusive} character with a power exponent $2/\beta$ for a finite time $n$ when $1\leq \beta<2$.}
Figure \ref{Fig: superdiffusion} shows the {supperdiffusive behavior} of $\langle |p_n-p_0|^q \rangle^{2/q}$
with power exponent $\eta=\frac{2}{\beta} = \frac{4}{3}$, {where $q=1<\beta$.}
\begin{figure}[H]
	\centering
	\includegraphics[width = .6\columnwidth]{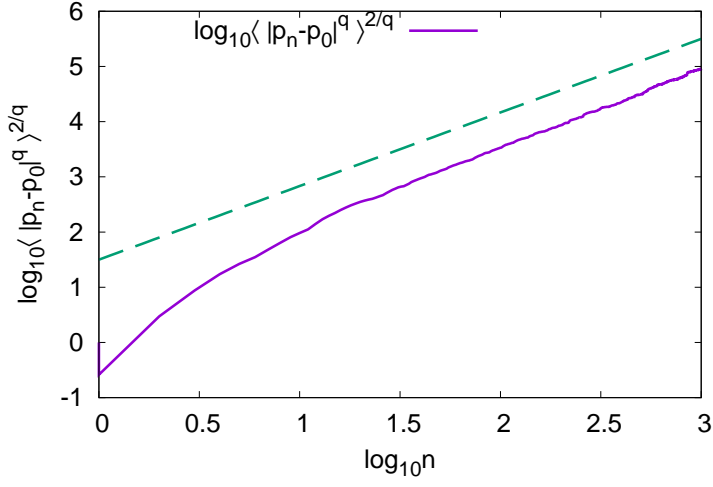}
	\caption{{Time development of $\langle |p_n-p_0|^q \rangle^{2/q}$} for $U_{3, \frac{1}{3}, \frac{3}{2}}$ (solid line) and a line with slope $\frac{4}{3}$ (broken line). The number of the ensemble is $1000$. As the initial conditions, $\{x_0\}$ {are uniformly distributed within $(-1,1)$}, and $\{p_0\}$ is set to be zero.}
	\label{Fig: superdiffusion}
\end{figure}

The power exponent $\eta$ can be continuously changed by changing $\beta$, which is the characteristic exponent of the class of superdiffusion.

\section{Critical exponent of the Lyapunov exponent}
We defined the critical exponent $\nu$ of the Lyapunov exponent as the supremum of $x$, which satisfies the equation below. 
\begin{equation}
\nu = \sup \left\lbrace x \in \mathbb{R}\left| \frac{\lambda}{|\alpha-\alpha_c|^x}\to 0~\mbox{as}~ \alpha \to \alpha_c\right.\right\rbrace,
\end{equation}
where $\alpha_c$ is the critical point of parameter $\alpha$.

According to this discussion, the behavior of $\gamma_{K, \alpha}$ as $\gamma_{K, \alpha} \to \infty$ or $\gamma_{K, \alpha} \to +0$
is the same as that in SGB transformations \cite{okubo2018universality}. Therefore,
for any $K\in \mathbb{N}\backslash{1}$,
\begin{equation}
\displaystyle \frac{1}{\gamma_{K, \alpha}}= O\left(\sqrt{1-|\alpha|}\right)
\end{equation}
 in the limit of {$\alpha \to 1-0$ or $\alpha \to -1+0$}.
For $K=2N+1$,
\begin{equation}
{\gamma_{2N+1,\alpha}} = O\left(\sqrt{|\alpha| - \frac{1}{(2N+1)^2}}\right)
\end{equation} 
in the limit of {$\alpha \to \frac{1}{(2N+1)^2}+0$ or $\alpha \to -\frac{1}{(2N+1)^2}-0$}. Moreover, 
for $K=2N$,
\begin{equation}
\gamma_{2N,\alpha} \sim \sqrt{|\alpha|} \label{Scaling 0<alpha<1/K, K=2N honbun}
\end{equation} in the limit of {$\alpha \to \pm0$}.

The invariant density functions in ESGB transformations are extended from the Cauchy distribution, corresponding to the invariant density of SGB transformations.
However, we proved that the critical exponents of the Lyapunov exponent for ESGB transformations
are the same as those for SGB transformations regardless of the power exponent of invariant density functions. 

The following theorem holds for the scaling behavior of the Lyapunov exponents $\lambda \sim b\left|\alpha- \alpha_{c}\right|^{\nu}$
for ESGB transformations.
\begin{theorem}\label{Scaling behavior U_{K, alpha, beta}}
	Assume that the parameters $(K, \alpha)$ are in Range B.
		For any $N\in \mathbb{N}\backslash\{1\}$, 
		{$\nu=\frac{1}{2}$ at $(K, \alpha) = (N, 1-0),(N, -1+0), 
		\left(2N+1, \frac{1}{(2N+1)^2}+0 \right), \left(2N+1, -\frac{1}{(2N+1)^2}-0\right)$, and $(2N, \pm 0)$}.
\end{theorem}

\begin{proof}
	If $x=|\cot\theta|^{\frac{1}{\beta}}\mbox{sgn}\left\lbrace \cot\theta\right\rbrace $, 
	\begin{equation}
	\begin{array}{rll}
	\displaystyle\frac{dx}{d\theta} &=& \displaystyle-\frac{|\cot\theta|^{\frac{1}{\beta}-1}}{\beta\sin^2\theta},\\
	\displaystyle\frac{dU}{d\theta} &=& \displaystyle-\frac{K|\alpha K|^{\frac{1}{\beta}}}{\beta}\frac{|\cot(K\theta)|^{\frac{1}{\beta}-1}}{\sin^2(K\theta)},\\
	\therefore \displaystyle\frac{dU}{dx} &=& \displaystyle\frac{dU}{d\theta}\frac{d\theta}{dx}\\
	&=& \displaystyle K|\alpha K|^{\frac{1}{\beta}}\frac{\sin^2\theta}{\sin^2(K\theta)}\left|\frac{\cot(K\theta)}{\cot\theta}\right|^{\frac{1}{\beta}-1}.
	\end{array}
	\end{equation}
	
	In this case, the Lyapunov exponent is calculated as
	\begin{equation}
	\lambda = \displaystyle \frac{1}{\pi}\int_{0}^{\pi}d\theta \log \left|K|\alpha K|^\frac{1}{\beta}\frac{\sin^2\theta}{\sin^2(K\theta)}
	\left|\frac{\cot(K\theta)}{\cot\theta}\right|^{\frac{1}{\beta}-1}\right|\cdot \frac{{\gamma_{K, \alpha}}|\cot\theta|^{\frac{\beta-1}{\beta}}}{\cot^2\theta+{\gamma_{K, \alpha}}^2}.
	\end{equation}

    If the topologically conjugate maps $\bar{S}_{K,\alpha}$ are considered, the Lyapunov exponent is finite and maximised when the invariant density is uniform ($\alpha = \frac{1}{K}$) \cite{lasota2008probabilistic,okubo2018universality}.
	For simplicity, we define $A= K^{1+\frac{1}{\beta}}\frac{\sin^2\theta}{\sin^2(K\theta)}\left|\frac{\cot(K\theta)}{\cot\theta}\right|^{\frac{1}{\beta}-1}$
	($A$ does not depend on ${\gamma_{K, \alpha}}$ or $\alpha$), and $\phi_1$ and $\phi_2$ are expressed as 
	\begin{equation}
	\phi_1(\theta, {\gamma_{K, \alpha}}, \beta) = \displaystyle \log||\alpha|^{\frac{1}{\beta}}A|\frac{{\gamma_{K, \alpha}}|\cot\theta|^{\frac{\beta-1}{\beta}}}{\cot^2\theta + {\gamma_{K, \alpha}}^2},    
	\end{equation}
	\begin{equation}
	\begin{array}{lll}
    \phi_2(\theta, z, \beta) &=& \displaystyle \log\left||\alpha|^{\frac{1}{\beta}}A\right|
	\frac{\frac{1}{{\gamma_{K, \alpha}}}|\cot\theta|^{\frac{\beta-1}{\beta}}}{\frac{1}{{\gamma_{K, \alpha}}^2}\cot^2\theta +1}\\
	&=& \displaystyle \log\left||\alpha|^{\frac{1}{\beta}}A\right|
	\frac{z|\cot\theta|^{\frac{\beta-1}{\beta}}}{z^2\cot^2\theta +1}. 
	\end{array}
	\end{equation}
	where $z \overset{\rm def}{=} \frac{1}{{\gamma_{K, \alpha}}}$.
	
	(i) In the limit of {$\alpha \to 1-0$ or $\alpha \to -1+0$} $(\gamma_{K,\alpha}\to \infty)$, the partial derivative of $\phi_2$ with respect to $z$ is considered 
	($\theta$ and $\beta$, which do not depend on ${\gamma_{K, \alpha}}$,
	and the interval $(a_n, a_{n+1}]$ is considered to correspond to SGB transformations \cite{okubo2018universality} to ensure that $\phi_2$ does not diverge).
	\begin{equation}
	\frac{\partial \phi_2}{\partial z} = \frac{1}{|\alpha|\beta}\frac{z|\cot|^{\frac{\beta-1}{\beta}}}{z^2\cot^2\theta+1}\frac{\partial |\alpha|}{\partial z}+
	\log\left||\alpha|^{\frac{1}{\beta}}A\right|\frac{|\cot\theta|^{\frac{\beta-1}{\beta}}(1-z^2\cot^2\theta)}{(z^2\cot^2\theta+1)^2}.
	\end{equation}
	When $K=2N$, $|\alpha| \sim 1- \frac{4N^2-1}{3}z^2$,   \cite{okubo2018universality}, and
	$|\alpha| \sim 1 - \frac{4N(N+1)}{3}z^2$ when $K=2N+1$ \cite{okubo2018universality}. In both cases,
	$\frac{\partial \phi_2}{\partial z}$ does not diverge in the limit of {$\alpha \to 1-0$ or $\alpha \to -1+0$} $({\gamma_{K, \alpha}} \to +\infty)$.
	According to the method described in a previous study \cite{okubo2018universality}, the following expression is obtained.
	\begin{equation}
	\begin{array}{lll}
	|\lambda_{K,\alpha, \beta}| &<& \infty,\\
	\lambda_{K, \alpha, \beta} &=& \displaystyle \frac{z}{\pi}\sum_{n=0}^{K-1}\int_{a_n}^{a_{n+1}}d\theta\left[
	\frac{\partial \phi_2}{\partial z}(\theta, 0, \beta)+ O(z)
	\right].
	\end{array}
	\end{equation}
	Because $\lambda_{K,\alpha, \beta}$ and $z$ are finite, 
	$\displaystyle  \frac{1}{\pi}\sum_{n=0}^{K-1}\int_{a_n}^{a_{n+1}}d\theta\left[
	\frac{\partial \phi_2}{\partial z}(\theta, 0, \beta)+ O(z)
	\right]$ is also finite. Moreover, $\frac{\partial \phi_2}{\partial z}(\theta, 0, \beta)$ does not depend on $z$.
	Therefore, in the limit of $z\to +0 (\gamma_{K, \alpha} \to \infty, {\alpha \to 1-0~\mbox{or}~\alpha \to -1+0})$, 
	\begin{equation}
	\lambda_{K, \alpha, \beta} = O(z) = O \left(\sqrt{1-|\alpha|}\right). \label{asymptotic critical exponent infty beta}
	\end{equation}
	Thus, 
	\begin{equation}
	   \nu = \frac{1}{2}. \label{eq:critical exponent nu_1}
	\end{equation}

	(ii) In the case of $K=2N+1$ and in the limit of $|\alpha|\to \frac{1}{(2N+1)^2}+0 (\gamma_{2N+1,\alpha} \to 0)$, 
	$\gamma_{2N+1,\alpha} = O\left(\sqrt{|\alpha|-\frac{1}{(2N+1)^2}}\right)$ holds.
	Consider the partial derivative of $\phi_1$ with respect to ${\gamma_{2N+1, \alpha}}$:
	\begin{equation}
	\displaystyle \frac{\partial \phi_1}{\partial {\gamma_{2N+1, \alpha}}}
	=  \displaystyle
	\frac{1}{|\alpha| \beta} \frac{{\gamma_{2N+1, \alpha}} \left|\cot\theta\right|^{\frac{\beta-1}{\beta}}}{\cot^2\theta+ {\gamma_{2N+1, \alpha}}^2} \frac{\partial |\alpha|}{\partial {\gamma_{2N+1, \alpha}}}
	+\log\left||\alpha|^{\frac{1}{\beta}}A\right|
	\frac{\left|\cot\theta\right|^{\frac{\beta-1}{\beta}}(\cot^2\theta -{\gamma_{2N+1, \alpha}}^2)}{(\cot^2\theta+{\gamma_{2N+1, \alpha}}^2)^2}.
	\end{equation}
	In this case, $\frac{\partial \phi_1}{\partial {\gamma_{2N+1, \alpha}}}$ does not diverge as $\alpha \to \frac{1}{(2N+1)^2}+0$. Therefore, 
	\begin{equation}
	\begin{array}{lll}
	|\lambda_{2N+1,\alpha, \beta}| &<& \infty,\\
	\lambda_{2N+1,\alpha,\beta} &=& \displaystyle \frac{1}{\pi}\int_0^\pi d\theta \phi_1(\theta,{\gamma_{2N+1, \alpha}},\beta)\\
	&=&\displaystyle \frac{1}{\pi}\sum_{n=0}^{{2N}}\int_{a_n}^{a_{n+1}}d\theta\left[
	\phi_1(\theta,0,\beta)+\frac{\partial \phi_1}{\partial \gamma_{2N+1,\alpha}}(\theta,0,\beta)\gamma_{2N+1,\alpha}+O(\gamma_{2N+1,\alpha})
	\right]\\
	&=&\displaystyle \frac{\gamma_{2N+1,\alpha}}{\pi}\sum_{n=0}^{{2N}}\int_{a_n}^{a_{n+1}}d\theta\left[
	\frac{\partial \phi_1}{\partial \gamma_{2N+1,\alpha}}(\theta,0,\beta)+O(\gamma_{2N+1,\alpha})
	\right].
	\end{array}
	\end{equation}
	Since $\lambda_{2N+1,\alpha,\beta}$ and $\gamma_{2N+1,\alpha}$ are finite, 
	$\displaystyle \frac{1}{\pi}\sum_{n=0}^{{2N}}\int_{a_n}^{a_{n+1}}d\theta\left[
	\frac{\partial \phi_1}{\partial {\gamma_{2N+1, \alpha}}}(\theta,0,\beta)+O({\gamma_{2N+1, \alpha}})\right]$ is finite.
	Moreover, $\frac{\partial \phi_1}{\partial {\gamma_{2N+1, \alpha}}}(\theta,0,\beta)$ does not depend on $\gamma_{2N+1,\alpha}$.
	Therefore, in the limit of $\gamma_{2N+1, \alpha} \to +0 ~(|\alpha| \to \frac{1}{(2N+1)^2}+0)$, 
	\begin{equation}
	\lambda_{2N+1,\alpha,\beta} = O(\gamma_{2N+1, \alpha}) = O\left(\sqrt{|\alpha|-\frac{1}{(2N+1)^2}}\right).
	\label{asymptotic critical exponent 0 beta}
	\end{equation}
	Therefore, $\nu$ is defined as
	\begin{equation}
	    \nu=\frac{1}{2}.\label{eq:critical exponent nu_2}
	\end{equation}
	
	(iii) In the case of $K=2N$ and in the limit of ${\alpha\to \pm0}$ $(\gamma_{2N,\alpha} \to 0)$, 
	$\gamma_{2N, \alpha} \sim \sqrt{|\alpha|}$.
	Therefore,
	\begin{equation}
	    |\lambda_{2N,\alpha,\beta}|<\infty
	\end{equation}
	\begin{equation}
	     \lambda_{2N,\alpha,\beta} 
	     = \displaystyle\frac{1}{\pi}\int_0^\pi d\theta \left[
	     \frac{1}{\beta}\gamma_{2N,\alpha}\log|\alpha|\frac{|\cot\theta|^{\frac{\beta-1}{\beta}}}{\cot^2\theta + {\gamma_{2N, \alpha}}^2}+
	     \gamma_{2N,\alpha}\log|A|\frac{|\cot\theta|^{\frac{\beta-1}{\beta}}}{\cot^2\theta + {\gamma_{2N, \alpha}}^2}
	     \right].
	\label{Eq: (iii) scaling relation of Lyapunov first order}
	\end{equation}

	For the first term of \eqref{Eq: (iii) scaling relation of Lyapunov first order}, according to \eqref{Scaling 0<alpha<1/K, K=2N honbun},
	in the limit of $|\alpha| \to 0$,
	\begin{equation}
	\frac{1}{\beta}\frac{\gamma_{2N,\alpha}\log|\alpha|}{|\alpha|^x}\frac{|\cot\theta|^{\frac{\beta-1}{\beta}}}{\cot^2\theta + {\gamma_{2N, \alpha}}^2}
	=
	\left\lbrace
	\begin{array}{ll}
	0, & x<\frac{1}{2}\\
	-\infty, & x \geq \frac{1}{2}
	\end{array}
	\right.
	\end{equation}
	
	For the second term of \eqref{Eq: (iii) scaling relation of Lyapunov first order}, according to \eqref{Scaling 0<alpha<1/K, K=2N honbun}, the following relation holds in the limit of $|\alpha| \to 0$.
	\begin{equation}
	\begin{array}{lll}
	\displaystyle \gamma_{2N,\alpha}\log|A|\frac{|\cot\theta|^{\frac{\beta-1}{\beta}}}{\cot^2\theta + {\gamma_{2N, \alpha}}^2} = O(\sqrt{|\alpha|})
	\end{array}
	\end{equation}
	
	Therefore, in the limit of $|\alpha| \to +0$, 
	\begin{equation}
	    \displaystyle \frac{\lambda_{{2N},\alpha,\beta}}{|\alpha|^x}=\frac{1}{\pi}\int_0^\pi \frac{\phi_1(\theta,{\gamma_{2N, \alpha}},\beta)}{|\alpha|^x}d\theta
	    \left\lbrace
	\begin{array}{ll}
	\to 0, & x<\frac{1}{2}\\
	\not\to 0, & x \geq \frac{1}{2}
	\end{array}
	\right.
	\label{asymptotic critical exponent 0 K=2N beta}
	\end{equation}
	Therefore, in the limit of $\gamma_{2N,\alpha}\to +0 ({\alpha\to \pm0})$,
	we obtained the critical exponent of the Lyapunov exponent as
	\begin{equation}
	   \nu = \frac{1}{2}. \label{eq:critical exponent nu_3}
	\end{equation}
	
	Theorem \ref{Scaling behavior U_{K, alpha, beta}} holds according to Equations \eqref{eq:critical exponent nu_1}, \eqref{eq:critical exponent nu_2}, and
	\eqref{eq:critical exponent nu_3}.
\end{proof}

\begin{figure}[h]
		\centering
		\hspace*{.5cm}
		\includegraphics[width = .6\columnwidth]{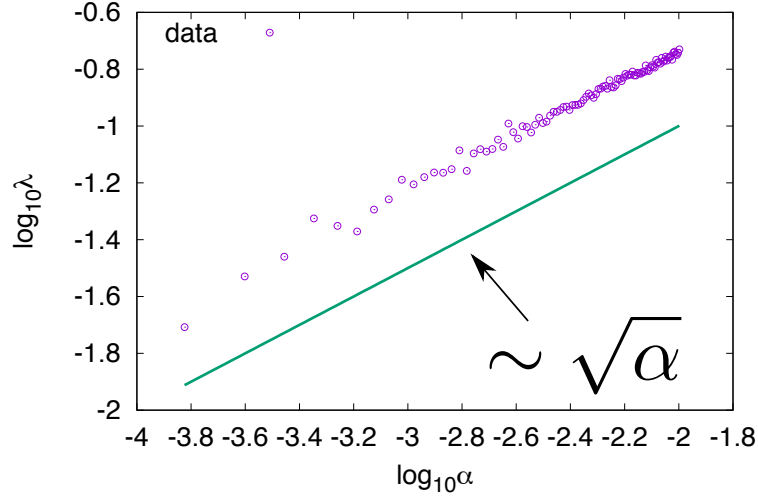}
		\caption{Scaling behavior of the Lyapunov exponent for $K=2 and \beta=1.5$. The initial point is $x_0=\sqrt{3}$. The number of iterations is $10^5$.}
		\label{Fig: Lyapunov scaling }
\end{figure}

\section{Conclusion}
In this study, we showed that the universality class of the chaos transition between the chaotic state and predictable state is more robust compared with the universality class of the power law index $\beta$ of stable laws in the framework of GCLT. To obtain this, ESGB transformations were introduced to derive the critical exponents $\nu$ of the Lyapunov exponent for a wide class of maps with ergodic invariant probability density function with the power {law} index $\beta+1$. The analytically obtained critical exponents $\nu(=1/2)$ are consistent with the critical exponents for the intermittent maps as predicted by Pomeau and Manneville \cite{pomeau1980intermittent}.
For ESGB transformations, through the applications to the topological conjugacy, we succeeded in extending the class of invariant density function from the Cauchy distribution (invariant density function of SGB transformations)
to $\rho(x) = \frac{1}{\pi}\frac{\gamma|x|^{\beta-1}}{|x|^{2\beta}+\gamma^2}$. 
Moreover, we proved that the ESGB transformations are \textit{exact} when parameters $(K, \alpha)$ are
in Range B regardless of $\beta$, which controls the degree of the fat tail.

In generalised Boole (GB) transformations \cite{umeno2016exact} (corresponding to $K=2$ for SGB transformations) whose invariant density function is the Cauchy distribution, we previously determined the critical exponent of the Lyapunov exponent as $\nu = \frac{1}{2}$ and $\alpha\to +0$ using the \textit{explicit} form of the Lyapunov exponent as a function of the parameters.
For the case of SGB transformations \cite{okubo2018universality}, we could not obtain the critical exponent $\nu$ at $(K, \alpha)=(2N, 0)$ for any $N\in \mathbb{N}$ because it is difficult to obtain the explicit form of the Lyapunov exponent for any $K=2N$. 

We obtained the critical exponent as $\nu = \frac{1}{2}$ at $(K, {\alpha})=(2N, {\pm0})$ by defining the critical exponent. Therefore, $\nu= \frac{1}{2}$
at {$(K, \alpha) = (N, 1-0),(N, -1+0), (2N+1, \frac{1}{(2N+1)^2}+0), (2N+1, -\frac{1}{(2N+1)^2}-0), \mbox{and} (2N, \pm0)$}
regardless of parameter $\beta$. It is noted that $\nu$ does not depend on $\beta$. This result is an extension of the result for GB transformations. Note that the critical exponents are independent of $\beta$, demonstrating the validity of the estimation conducted by Pomeau and Manneville \cite{pomeau1980intermittent} 
for a countably infinite number of maps whose invariant density belongs to the Cauchy and wider classes. 
In other words, the universality of the critical exponent $\nu=\frac{1}{2}$ is applicable over the wider class. 

In future investigations, we intend to identify a possible scaling relation between $\nu$ and other critical exponents because a certain scaling relation \cite{fisher1967theory} could unify the universality class of phase transition. 
We expect our rigorous analysis to help enhance the understanding of chaos and critical phenomena of transition to chaos and provide novel insights regarding chaos in physics.


\appendix
\section{Proof that the map $U_{K,\alpha, \beta}$ is a $K$-to-one map}
In this section, we prove that the map $U_{K,\alpha, \beta}$ is a $K$-to one map.
The number of inverse images $x$ is exactly $K$ if $X$ is given. 
\begin{equation}
    X = \left|\alpha K F_K\left(|x|^\beta \mbox{sgn}(x)\right)\right|^\frac{1}{\beta} 
    \mbox{sgn}\left\lbrace \alpha F_K\left(|x|^\beta \mbox{sgn}(x)\right) \right\rbrace.
    \label{Eq:Apendix X}
\end{equation}

\noindent(i) In the case $\alpha F_K\left(|x|^\beta\mbox{sgn}(x)\right)\geq 0,$

Eq.(\ref{Eq:Apendix X}) is transformed as 
\begin{equation*}
    \begin{array}{lll}
        \displaystyle X &=&  \left\lbrace\alpha K F_K\left(|x|^\beta \mbox{sgn}(x)\right)\right\rbrace^\frac{1}{\beta}\\
        \displaystyle X^\beta &=&  \alpha K F_K\left(|x|^\beta \mbox{sgn}(x)\right)\\
        \displaystyle \frac{X^\beta}{\alpha K} &=& F_K\left(|x|^\beta \mbox{sgn}(x)\right) =F_K(\cot \theta) = \cot(K\theta).
    \end{array}
\end{equation*}
The properties of $\cot(K\theta)$ indicate that the number of inverse images of $\cot(K\theta)$ is $K$. Then,
\begin{equation*}
    \begin{array}{cll}
        \displaystyle K\theta & =& \displaystyle\mbox{arccot}\left(\frac{X^\beta}{\alpha K}\right) + j\pi \mod{\pi}, j=0,\cdots\\
        \displaystyle\theta &=& \displaystyle\frac{1}{K}\mbox{arccot}\left(\frac{X^\beta}{\alpha K}\right) + \frac{j}{K}\pi \mod{\pi}, j=0,\cdots,K-1.\\
        \displaystyle\therefore |x|^\beta \mbox{sgn}(x) &=& \displaystyle\cot\left\lbrace \frac{1}{K}\mbox{arccot}\left(\frac{X^\beta}{\alpha K}\right) + \frac{j}{K}\pi\right\rbrace, j=0,\cdots, K-1.
    \end{array}
\end{equation*}
As $|x|^\beta \mbox{sgn}(x)$ is one-to-one, the number of points $x$ each of which is an inverse image of $X$ is exactly $K$.

\noindent(ii) In the case $\alpha F_K\left(|x|^\beta\mbox{sgn}(x)\right)< 0,$
we have
\begin{equation*}
    -\frac{(-X)^\beta}{\alpha K} = F_K\left(|x|^\beta \mbox{sgn}(x)\right) = F_K(\cot \theta) = \cot(K\theta).
\end{equation*}
This case can be discussed similarly to the discussion regarding (i).

\bibliographystyle{ptephy}
\bibliography{PREESGB2}

\end{document}